\documentclass[a4paper,submission,copyright,creativecommons]{eptcs}
\providecommand{\event}{GandALF 2016} 
\usepackage{breakurl}             

\usepackage{latexsym}
\usepackage{url}
\usepackage{amssymb}
\usepackage{amsmath}
\usepackage{amsthm}
\usepackage{amsfonts}
\usepackage{tikz}
\usepackage{xspace}
\usepackage{paralist}

\theoremstyle{plain}
\newtheorem{proposition}{Proposition}
\newtheorem{theorem}[proposition]{Theorem}
\newtheorem{corollary}[proposition]{Corollary}

\newtheorem{lemma}[proposition]{Lemma}

\newcommand{\mmodels}{\Vdash}

\newcommand{\diax}{{\ensuremath{\Diamond_\alpha}\xspace}}
\newcommand{\boxx}{{\ensuremath{\Box_\alpha}\xspace}}

\newcommand{\Rx}{\mathrel{R_\alpha}}

\newcommand{\KnBool}{\ensuremath{\operatorname{\mathbf{K_N^{Bool}}}}\xspace}
\newcommand{\Kn}{\ensuremath{\operatorname{\mathbf{K_N}}}\xspace}
\newcommand{\KnHorn}{\ensuremath{\operatorname{\mathbf{K_N^{Horn}}}}\xspace}
\newcommand{\KnHornBox}{\ensuremath{\operatorname{\mathbf{K_N^{Horn,\Box}}}}\xspace}
\newcommand{\KnHornDia}{\ensuremath{\operatorname{\mathbf{K_N^{Horn,\Diamond}}}}\xspace}
\newcommand{\KnKrom}{\ensuremath{\operatorname{\mathbf{K_N^{Krom}}}}\xspace}
\newcommand{\KnKromBox}{\ensuremath{\operatorname{\mathbf{K_N^{Krom,\Box}}}}\xspace}
\newcommand{\KnKromDia}{\ensuremath{\operatorname{\mathbf{K_N^{Krom,\Diamond}}}}\xspace}
\newcommand{\KnCore}{\ensuremath{\operatorname{\mathbf{K_N^{core}}}}\xspace}
\newcommand{\KnCoreBox}{\ensuremath{\operatorname{\mathbf{K_N^{core,\Box}}}}\xspace}
\newcommand{\KnCoreDia}{\ensuremath{\operatorname{\mathbf{K_N^{core,\Diamond}}}}\xspace}

\newcommand{\HS}{\ensuremath{\mathbf{HS}}\xspace}

\newcommand{\LTL}{\ensuremath{\mathbf{LTL}}\xspace}
\newcommand{\QBF}{\ensuremath{\mathbf{QBF}}\xspace}
\newcommand{\Gr}{\ensuremath{\mathbf{GR(1)}}\xspace}

\newcommand{\NP}{\textsc{NP}}
\newcommand{\Pt}{\textsc{P}}

\newcommand{\NLOG}{\textsc{NLogSpace}}
\newcommand{\PSpace}{\textsc{PSpace}}

\title{On the Expressive Power of Sub-Propositional Fragments of Modal Logic\thanks{The authors acknowledge the support from the Italian INdAM -- GNCS Project 2016 \emph{`Logic, automata and games for self-adaptive systems'} (D. Bresolin and G. Sciavicco), and the Spanish Project \emph{TIN15-70266-C2-P-1} (E. Mu\~noz-Velasco).}}

\author{Davide Bresolin
\institute{Department of Computer Science and Engineering \\  University of Bologna (Italy)}
\email{davide.bresolin@unibo.it}
\and
Emilio Mu\~noz-Velasco
\institute{Department of Applied Mathematics\\ University of M\'alaga (Spain)}
\email{emilio@ctima.uma.es}
\and
Guido Sciavicco
\institute{Department of Mathematics and Computer Science\\ University of Ferrara (Italy)}
\email{guido.sciavicco@unife.it}
}

\def\titlerunning{Expressive Power of Sub-Propositional Multi-Modal Logics}
\def\authorrunning{D. Bresolin, E. Mu\~noz-Velasco \& G. Sciavicco}

\begin{document}

\maketitle

\begin{abstract}
Modal logic is a paradigm for several useful and applicable formal systems in computer science.
It generally retains the low complexity of classical propositional logic, but notable exceptions exist in the domains of description, temporal, and spatial logic, where the most expressive formalisms have a very high complexity or are even undecidable.
In search of computationally well-behaved fragments, clausal forms and other sub-propositional restrictions of temporal and description logics have been recently studied.
This renewed interest on sub-propositional logics, which mainly focus on the complexity of the various fragments, raise natural questions on their the relative expressive power, which we try to answer here for the basic multi-modal logic \Kn. We consider the Horn and the Krom restrictions, as well as the combined restriction (known as the core fragment) of modal logic, and, orthogonally, the fragments that emerge by disallowing boxes or diamonds from positive literals. We study the problem in a very general setting, to ease transferring our results to other meaningful cases.
\end{abstract}

\section{Introduction}\label{sec:intro}

The usefulness and the applicability of modal logic is well-known and accepted. Propositional modal logic generally retains the decidability of the satisfiability problem of classical propositional logic, but extends its language with {\em existential modalities} ({\em diamonds}, to express {\em possibility}) and their {\em universal} versions ({\em boxes}, to express {\em necessity}), allowing one to formalize a much wider range of situations.
To simply cite a few, modal logic has been applied not only to philosophical reasoning (e.g., epistemological, or metaphysical reasoning - see~\cite[Chapter~1]{modal-logic} for an historical perspective), but also to computer science, being paradigmatic of the whole variety of description logics~\cite{Baader:2003:DLH:885746}, temporal logics~\cite{temporal_logic_foundations}, and spatial logics~\cite{spatial_logic_handbook}.

Until very recently, clausal fragments of modal logic has received little or no attention, with the exception of a few works which are limited to the Horn fragment~\cite{ChenLin94,del1987note,nguyen2004complexity}. An inversion in this tendency is mainly due to the newborn interest in sub-propositional fragments of temporal description logics~\cite{DBLP:journals/tocl/ArtaleKRZ14}, temporal logics~\cite{DBLP:conf/lpar/ArtaleKRZ13}, and interval temporal logics~\cite{artale2015,DBLP:conf/jelia/BresolinMS14}. Such results, which mainly concern the complexity of various sub-propositional fragments of description and temporal logics raise natural questions on their the relative expressive power, which we try to answer here in a very general form. 

There are two standard ways to weaken the classical propositional language based on the clausal form of formulas:
the {\em Horn fragment}, that only allows clauses with at most one positive literal~\cite{horn}, and the {\em Krom fragment}, that only allows clauses with at most two (positive and negative) literals~\cite{krom}. The {\em core fragment} combines both restrictions. Orthogonally, one can restrict a modal language in clausal form by disallowing either diamonds or boxes in positive literals, obtaining weaker fragments that we call, respectively, the {\em box fragment} and {\em diamond fragment}. By combining these two levels of restrictions, one may obtain several sub-propositional fragments of modal logic, and, by extensions, of description, temporal, and spatial logics. 
The interest in such fragments is originated by the quest of computationally well-behaved logics, and by the observation that meaningful statements can be still expressed under the sub-propositional restrictions. The satisfiability problem for classical propositional Horn logic is $\Pt$-complete~\cite{cookhorn}, while for classical propositional Krom logic (also known as the 2-SAT problem) it is $\NLOG$-complete~\cite{papa}, and the same holds for the core fragment. Interestingly enough, the satisfiability problem for quantified propositional logic (\QBF), which is \PSpace-complete in its general form, becomes \Pt\ when formulas are restricted to binary (Krom) clauses~\cite{asp79}. 

Sub-propositional modal logic has been studied mainly under the Horn restriction. The basic modal logic $\mathbf K$, which is \PSpace-complete, remains so under the Horn restriction, but the satisfiability problem for other cases becomes computationally easier, such as $\mathbf{S5}$, which goes from being \NP-complete to \Pt-complete~\cite{ChenLin94,del1987note}. In~\cite{DBLP:conf/lpar/ArtaleKRZ13,ChenLin93}, the authors study different sub-propositional fragments of Linear Temporal Logic (\LTL). By excluding the Since and Until operators from the language, and keeping only the Next/Previous-time operators and the Future and Past box modalities, it is possible to prove that the Krom and core fragments are \NP-hard, while the Horn fragment is still \PSpace-complete (the same complexity of the full language). Moreover, the complexity of the Horn, Krom, and core fragments without Next/Previous-time operators range from \NLOG\ (core), to \Pt\ (Horn), to \NP-hard (Krom). Where only a universal (anywhere in time) modality is allowed their complexity is even lower (from \NLOG\ to \Pt). Temporal extensions of the description logic DL-Lite have been studied under similar sub-propositional restrictions, and similar improvements in the complexity of various problems have been found~\cite{DBLP:journals/tocl/ArtaleKRZ14}. Sub-propositional fragments of the undecidable interval temporal logic \HS~\cite{interval_modal_logic}, have also been studied. The Horn, Krom, and core restrictions of \HS\ are still undecidable~\cite{DBLP:conf/jelia/BresolinMS14}, but weaker restrictions have shown positive results. In particular, the Horn fragment of \HS\ without diamonds becomes $\Pt$-complete in two interesting cases~\cite{artale2015,bresolin:2016}: when it is interpreted over dense linear orders, and when the semantics of its modalities becomes reflexive. On the bases of these results, sub-propositional interval temporal extensions of description logics have been introduced in~\cite{artale2015}. 
Other clausal forms of temporal logics, not included in the above classification, have been developed 
to synthesize systems from logical specifications, as the logical counterpart of deterministic automata. 
The most relevant example is the fragment \Gr\ of \LTL~\cite{DBLP:journals/jcss/BloemJPPS12}, for which synthesis is exponentially more efficient than for full \LTL.


The purpose of this paper is to consider sub-propositional fragments of the multi-modal logic \Kn, and study their relative expressive power in a systematic way. We consider two different notions of relative expressive power for fragments of modal logic, and we provide several results that give rise to two different hierarchies among them, leaving only a few open problems. To the best of our knowledge, this is the first work where sub-Krom and sub-Horn fragments of \Kn have been considered.


\section{Preliminaries}\label{sec:prel}

Let us fix a unary modal {\em similarity type} as the set $\tau$ of modalities $\alpha_1,\alpha_2,\ldots,\alpha_N\in\tau$, and a denumerable set
$\mathcal P$ of propositional letters. The {\em modal language} \Kn associated to $\tau$ and $\mathcal P$ contains all and only the formulas generated by the following grammar:

\begin{equation}
\varphi::=\top\mid p\mid \neg\varphi\mid \varphi\vee\varphi\mid \diax\varphi\mid \boxx\varphi,\label{boolgrammar}
\end{equation}

\noindent where $p\in\mathcal P$, and $\alpha\in\tau$ labels the {\em diamond} $\diax$ and {\em box} $\boxx$. Other classical operators, such as $\rightarrow$ and $\wedge$, can be considered as abbreviations. A {\em Kripke} $\tau$-{\em frame} is a relational $\tau$-structure $\mathcal F=(W,\{R\}_{\alpha\in\tau})$, where the elements of $W\neq\emptyset$ are called {\em possible worlds}, and, for each $\alpha\in\tau$, $R_\alpha\in W\times W$ is an {\em accessibility} relation. A {\em Kripke structure} over the $\tau$-frame $\mathcal F$ is a pair $M=(\mathcal F,V)$, where $V:W\rightarrow 2^\mathcal P$ is an {\em evaluation function}, and we say that $M$ {\em models} $\varphi$ at the world $w$, denoted by $M,w\mmodels\varphi$, if and only if:

\medskip

\begin{compactitem}
\item $\varphi=\top$;
\item $p\in V(w)$, if $\varphi = p$;
\item $M,w\not\mmodels\psi$, if $\varphi= \neg\psi$;
\item $M,w\mmodels\psi$ or $M,w\mmodels\xi$, if $\varphi= \psi\vee\xi$;
\item There exists $v$ such that $wR_\alpha v$ and $M,v\mmodels\psi$, if $\varphi= \diax\psi$.
\item For every $v$ such that $wR_\alpha v$, it is the case that $M,v\mmodels\psi$, if $\varphi= \boxx\psi$.
\end{compactitem}

\medskip

\noindent In this case, we say that $M$ is a {\em model} of $\varphi$; in the following, we (improperly) use the terms models and structures as synonyms. 

In order to define sub-propositional fragments of \Kn we start from the {\em clausal form} of \Kn-formulas, whose building blocks are the {\em positive literals}:

\begin{equation}
\lambda  ::= \top \mid p\mid \diax \lambda \mid \boxx \lambda,\label{poslit}
\end{equation}

\noindent and we say that $\varphi$ is in {\em clausal form} if it can be generated by the following grammar:

\begin{equation}
\varphi  ::= \lambda \mid \neg\lambda\mid \nabla (\neg\lambda_1\vee\neg\lambda_2\vee\ldots\vee\neg\lambda_n\vee\lambda_{n+1}\lor\lambda_{n+2}\vee\ldots\vee\lambda_{n+m})\mid \varphi \land \varphi,\label{clause}
\end{equation}

\noindent where $\nabla=\underbrace{\Box_{\alpha_i}\Box_{\alpha_j}\ldots}_{s}$ and $s\ge 0$. Sometimes, we write clauses in their implicative form:

\vspace{-0.5\baselineskip}
\begin{equation*}
\nabla(\lambda_1\land\ldots\land\lambda_n\rightarrow\lambda_{n+1}\vee\ldots\vee\lambda_{n+m}),
\end{equation*}
 
\noindent and we use $\bot$ as a shortcut for $\neg\top$. By $md(\lambda)$ we mean the {\em modal depth} of $\lambda$, that is, the number of boxes and diamonds in $\lambda$. {\em Sub-propositional} fragments of \Kn can be now defined by constraining the cardinality and the structure of clauses: the fragment of \Kn in clausal form where each clause in (\ref{clause}) is such that $m\le 1$ is called {\em Horn} fragment, and denoted by \KnHorn, and when each clause is such that $n+m\le 2$ it is called {\em Krom} fragment, and it is denoted by \KnKrom. When both restrictions apply we denote the resulting fragment, the {\em core} fragment, by \KnCore. We use \KnBool instead of \Kn\ to highlight that no restrictions apply. It is also interesting to study the fragments that can be obtained from both the Horn and the Krom fragments by disallowing, respectively, the use of \boxx or \diax in positive literals. In this way, the fragment of \KnHorn obtained by eliminating the use of diamonds (resp., boxes) in (\ref{poslit}) is denoted by \KnHornBox (resp., \KnHornDia). By applying the same restrictions to \KnKrom and \KnCore, one obtains the pair \KnKromDia and \KnKromBox from the former, and the pair \KnCoreDia and \KnCoreBox, from the latter. All such sub-Horn, sub-Krom, and sub-core fragments are generally called {\em box} and {\em diamond} fragments. 

\medskip

It should be noted that in the literature there is no unified definition of the different modal or temporal sub-propositional logics. Our definition follows the one by Nguyen~\cite{nguyen2004complexity}, with a notable difference: while the definition of clauses is the same, we choose a more restrictive definition of what is a formula. Hence, a formula of \KnHorn by our definition is also a Horn formula by ~\cite{nguyen2004complexity}, but not vice versa. However, since every Horn formula by~\cite{nguyen2004complexity} can be transformed into a conjunction of Horn clauses, the two definitions are equivalent. The definition of~\cite{ChenLin94,del1987note} is equivalent to that of Nguyen, and hence to our own. Other approaches force clauses to be quantified using a {\em universal} modality that asserts the truth of a formula in every world of the model. The universal modality is either assumed in the language~\cite{DBLP:conf/lpar/ArtaleKRZ13} or it is definable using the other modalities~\cite{bresolin:2016,DBLP:conf/jelia/BresolinMS14}, but the common choice in the literature of modal (non-temporal) logic is simply excluding the universal modality. Our results hold in either case: when the universal modality is present (as part of the language or defined), and clauses are always universally quantified, they become even easier to prove.

There are many ways to compare the expressive power of different modal languages. In our context, two different concepts of expressive equivalence arise naturally.
The first one, that we call \emph{weak expressivity}, compares formulas (and models) with the same set of propositional letters. More formally, given two modal logics $\mathbf L$ and $\mathbf L'$ interpreted in the same class of relational frames $\mathcal{C}$, we say that $\mathbf L'$ is {\em weakly at least as expressive as} $\mathbf L$ if, fixed a propositional alphabet $\mathcal P$,
there exists an effective translation $(\cdot)'$ from $\mathbf L$ to
$\mathbf L'$ such that for every model
$M$ in $\mathcal{C}$, world $w$ in $M$, and formula $\varphi$ of $\mathbf L$,
we have $M, w \mmodels \varphi$ if and only if  $M, w \mmodels \varphi'$. We denote this situation with $\mathbf L\preceq_{\mathcal C}^w\mathbf L'$, and we omit $\mathcal C$ if it is clear from the context. The second notion, that we call \emph{strong expressivity}, allows the translations to use a finite number of new propositional letters, and can be formally defined as follows. For every model $M=(\mathcal F,V)$ based on the set of propositional letters $\mathcal P$ and every $\mathcal P'\supseteq\mathcal P$, we say that the model $M^{\mathcal P'}=(\mathcal F,V^{\mathcal P'})$ based on $\mathcal P'$ is a {\em extension} of $M$ if $V|_{\mathcal P}=V'|_{\mathcal P}$.
Then, we say that $\mathbf L'$ is {\em at least as expressive as} $\mathbf L$ if there exists an effective translation $(\cdot)'$ that transforms any $\mathbf L$-formula $\varphi$ written in the alphabet $\mathcal P$ into a $\mathbf L'$-formula written in a suitable alphabet $\mathcal P'\supseteq\mathcal P$, such that for every model $M$ in $\mathcal{C}$ and world $w$ in $M$, we have that $M, w \mmodels \varphi$ if and only if there exists an extension $M'$ of $M$ such that $M', w \mmodels \varphi'$.  We denote this situation with $\mathbf L\preceq_{\mathcal C}\mathbf L'$.
Now, we can say that $\mathbf L$ and $\mathbf L'$ are {\em weakly equally expressive} if $\mathbf L\preceq_{\mathcal C}^w\mathbf L'$ and $\mathbf L'\preceq_{\mathcal C}^w\mathbf L$, and they are {\em equally expressive}  if $\mathbf L\preceq_{\mathcal C}\mathbf L'$ and $\mathbf L'\preceq_{\mathcal C}\mathbf L$; in the former case we write $\mathbf L\equiv^w\mathbf L'$, and in the latter case we write $\mathbf L\equiv\mathbf L'$. Finally, $\mathbf L$ is {\em weakly less expressive than} $\mathbf L'$ if $\mathbf L\preceq_{\mathcal C}^w\mathbf L'$ and $\mathbf L\not\equiv_{\mathcal C}^w\mathbf L'$, and $\mathbf L$ is {\em less expressive than} $\mathbf L'$ if $\mathbf L\preceq_{\mathcal C}\mathbf L'$  and $\mathbf L\not\equiv_{\mathcal C}\mathbf L'$; in the former case we write $\mathbf L\prec_{\mathcal C}^w\mathbf L'$, while in the latter one we write $\mathbf L\prec_{\mathcal C}\mathbf L'$. Clearly, two logics can be equally expressive and not weakly so,
but not the other way around.

Given $\mathbf L$ and $\mathbf L'$ such that $\mathbf L$ is a syntactical fragment of $\mathbf L'$, in order to prove that $\mathbf L$ is (weakly) less expressive than $\mathbf L'$ we show a formula $\psi$ that can be written in $\mathbf L'$ but not in $\mathbf L$. To this end we proceed by contradiction, assuming that a translation $\varphi\in\mathbf L$ does exist, and by building a model for $\psi$ that is not (and, in the case of strong relative expressiveness, cannot be extended to) a model of $\varphi$, following three different strategies: we modify the labeling (Theorem~\ref{th:horn_vs_bool} and Theorem~\ref{th:krom_vs_bool}), we modify the structure (Theorem~\ref{th:krombox_vs_krom} and Theorem~\ref{th:kromdia_vs_boxdia}), or we exploit a property of $\mathbf L'$ that $\mathbf L$ does not possess (Theorem~\ref{th:hornbox_vs_boxdia} and Theorem~\ref{th:horndia_vs_boxdia}). The two different levels that emerged from the above discussion give rise to two different hierarchies: \begin{inparaenum}[\it (i)] \item a {\em weak} hierarchy that compares fragments within the same propositional alphabet, and \item a {\em strong} hierarchy that takes into account any finite extension of the propositional alphabet. \end{inparaenum}

Adding new propositional letters to facilitate translations from a fragment to another is a common practice, for example, to prove that every $n$-ary clause in propositional logic can be transformed into an equi-satisfiable set of ternary clauses. 
In this sense, it can be argued that the weak hierarchy is less general; nonetheless, both the weak and the strong hierarchies contribute to the comprehension of the relative expressive power of sub-propositional fragments. Indeed, both notions have been already studied under different names~\cite{DBLP:conf/dlog/KonevLWZ15}: our weak hierarchy captures the notion of {\em equivalently rewritability}, while the strong one captures the notion of {\em model-conservative rewritability}.



\section{Horn, Krom, and Core Fragments}\label{sec:kn_expr}

In this section, we study the relative expressive power of the basic multi-modal logic \KnBool and its sub-propositional fragments with both boxes and diamonds. From now on, we focus on the class of all relational frames, and we omit it from the notation. We start by comparing the Horn fragment \KnHorn with the full propositional language. 

\begin{theorem}\label{th:horn_vs_bool}
$\KnHorn\prec^w\KnBool$.
\end{theorem}

\begin{proof}
Since $\KnHorn$ is a syntactical fragment of $\KnBool$, we know that $\KnHorn\preceq^w\KnBool$. It remains to be proved that there exists a formula that belongs to $\KnBool$ and that cannot be translated to $\KnHorn$ within the same propositional alphabet. Consider the \KnBool-formula
$$\psi\equiv p \lor q,$$

\noindent and suppose, by contradiction, that there exists a \KnHorn-formula $\varphi$ such that for every model $M$  over the propositional alphabet $\{p,q\}$, and every world $w$, we have that $M, w\mmodels \psi$ if and only if $M, w \mmodels \varphi.$ We can assume that $\varphi = \varphi_1 \land \ldots \land \varphi_l$, where each $\varphi_i$ is a positive literal, the negation of a positive literal, or a Horn clause. To simplify our argument, if $\varphi_i=\lambda$ (resp., $\varphi_i=\neg\lambda$) we shall think of it as the clause $(\top\rightarrow\lambda)$ (resp., $(\lambda\rightarrow\bot)$). Let us denote by $C(\varphi_i)$ the set of propositional letters that occur in the consequent of $\varphi_i$: clearly, $C(\varphi_i)$ is always a singleton, or it is the empty set. Now, consider a model $M = \langle \mathcal F, V\rangle$, where $\mathcal F$ is based on the set of worlds $W$, and let $w\in W$ be a world such that $M,w\not\mmodels\psi$. Such a model must exist since $\psi$ is not a tautology. Since $\varphi$ is a conjunction of Horn clauses, we have that there must exist at least one clause $\varphi_i = \nabla(\lambda_1 \land \ldots \land \lambda_n \rightarrow \lambda)$ such that $M, w \not\mmodels \varphi_i$. Hence, there must exist a world $w'$ such that $M, w' \mmodels \lambda_1 \land \ldots \land \lambda_n$ but $M, w' \not\mmodels \lambda$. At this point, only three cases may arise (since we are in a fixed propositional alphabet):

\medskip

\begin{compactitem}
\item $C(\varphi_i)=\{p\}$. In this case, we can build a new model $M' = \langle \mathcal F,V'\rangle$ such that:
$$V'(p)=V(p)~\mbox{and}~V'(q)=W.$$

\noindent Since $q$ holds on every world of the model, we have that $M'$ satisfies $\psi$ on every world, and, in particular, on $w$. However, being $\lambda_1,\ldots,\lambda_n$ positive literals, they are true on $M'$ whenever they were true on $M$, which means that $M',w'\mmodels \lambda_1 \land \ldots \land \lambda_n$. Now, consider the positive literal $\lambda$, we want to prove that, for each world $v\in W$,
$M,v\not\models\lambda$ implies $M',v\not\models\lambda$. We reason by induction on  $md(\lambda)$. If $md(\lambda)=0$, then $\lambda=p$; since $M$ and $M'$ agree on the valuation of the proposition $p$, we have the claim. Suppose, now, that $md(\lambda)>0$. Clearly, $\lambda=\diax\lambda'$ or $\lambda=\boxx\lambda'$; assume, first, that $\lambda=\diax\lambda'$. If $M,v\not\mmodels\diax\lambda'$, then, for every $t\in W$ such that $vR_\alpha t$, we have that $M,t\not\mmodels\lambda'$; by inductive hypothesis, for every $t\in W$ such that $vR_\alpha t$, we have that $M',t\not\mmodels\lambda'$, proving that, in fact, $M',v\not\mmodels\diax\lambda'$. Now, assume that $\lambda=\boxx\lambda'$. If $M,v\not\mmodels\boxx\lambda'$, then for some $t\in W$ such that $vR_\alpha t$ we have that $M,t\not\mmodels\lambda'$; by inductive hypothesis, $M',t\not\mmodels\lambda'$, which implies that $M',v\not\mmodels\boxx\lambda'$. Since $M,w'\not\mmodels\lambda$, the above argument proves that $M',w'\not\models\lambda$, which means that $M',w'\not\models\varphi_i$. This means that $M',w\mmodels\psi$ and $M',w\not\mmodels\varphi$, contradicting the fact that $\psi$ is a translation of $\varphi$.

\item $C(\varphi_i)=\{q\}$. In this case one can apply the same argument as before, by simply switching the roles of $p$ and $q$.

\item $C(\varphi_i)=\emptyset$. In this case, we can build a new model $M' = \langle \mathcal F,V'\rangle$ such that:
$$V'(p) = V'(q) = W.$$

\noindent Since $p$ and $q$ hold on every world of the model, we have that $M'$ satisfies $p\vee q$ everywhere, and, in particular, on $w$. However, since the truth of $\lambda$ does not depend on the valuations of the propositional letters, we have that, as before, $M',w' \mmodels \lambda_1 \land \ldots \land \lambda_n$ but $M',w' \not\models \lambda$, from which we can conclude that $M',w \not\mmodels \varphi$.
\end{compactitem}

\medskip

\noindent Therefore, $\varphi$ cannot exist, and this means that $\psi$ cannot be expressed in \KnHorn within the same propositional alphabet. So, the claim is proved.
\end{proof}

\medskip

Now, we turn our attention to the relationship between $\KnKrom$ and $\KnBool$. 

\begin{theorem}\label{th:krom_vs_bool}
$\KnKrom \prec^w\KnBool$.
\end{theorem}

\begin{proof}
Since $\KnKrom$ is a syntactical fragment of $\KnBool$, we know that $\KnKrom\preceq^w\KnBool$. It remains to be proved that there exists a formula that belongs to $\KnBool$ and that cannot be translated to $\KnKrom$ within the same propositional alphabet. Now, consider the \KnBool-formula

$$\psi \equiv p \land q \rightarrow r,$$

\noindent and suppose, by contradiction, that there exists a \KnKrom-formula $\varphi$, written in the propositional alphabet $\{p,q,r\}$, such that for every model $M$ and every world $w$ we have that $M, w\mmodels \psi$ if and only if $M,w \mmodels \varphi.$ As before, we can assume that $\varphi = \varphi_1 \land \ldots \land \varphi_l$; as in Theorem~\ref{th:horn_vs_bool}, if $\varphi_i$ is a literal, we treat it as a special clause. Let us denote by $P(\varphi_i)$ the set of propositional letters that occur in $\varphi_i$. Now, consider a model $M = \langle \mathcal F, V\rangle$, where $\mathcal F$ is based on the set of worlds $W$, and let $w\in W$ be a world such that $M,w\not\mmodels\psi$. Such a model must exist since $\psi$ is not a tautology. Since $\varphi$ is a conjunction of Krom clauses, we have that there must exist at least one clause $\varphi_i = \nabla(\lambda_1 \vee \lambda_2)$ such that $M, w \not\mmodels \varphi_i$. Hence, there must exist a world $w'$ such that $M, w' \not\mmodels (\lambda_1 \vee \lambda_2)$. At this point, three cases may arise (since we are in a fixed propositional alphabet, and we deal with clauses at most binary):

\medskip

\begin{compactitem}

\item $P(\varphi_i) \subseteq \{p,q\}$. In this case, we can build a new model $M' = \langle \mathcal F,V'\rangle$ such that:

 $$V'(p) = V(p),\ V'(q) = V(q),~\mbox{and}~V'(r) = W.$$

 \noindent Since $r$ holds on every world of the model, we have that $M'$ satisfies $\psi$ everywhere, and in particular on $w$. However, since the valuation of $p$ and $q$ are the same of $M$, and since the relational structure has not changed, we have that $M',w' \not\models \lambda_1 \lor \lambda_2$, from which we can conclude that $M',w \not\mmodels \nabla ( \lambda_1 \lor \lambda_2)$ and thus that $w$ do not satisfy $\varphi$.
	
\item $P(\varphi_i) \subseteq \{p,r\}$. In this case, we can build a new model $M' = \langle \mathcal F,V'\rangle$ such that:

$$V'(p) = V(p),\ V'(r) = V(r),~\mbox{and}~V'(q) = \emptyset.$$

\noindent Since $q$ is false on every world of the model, we have that $M'$ satisfies $\psi$ everywhere, and in particular on $w$. However, since the valuation of $p$ and $r$ are the same of $M$, and since the relational structure has not changed, we have that $M',w' \not\models \lambda_1 \lor \lambda_2$, from which we can conclude that $M',w \not\mmodels \nabla ( \lambda_1 \lor \lambda_2)$ and thus that $w$ do not satisfy $\varphi$.

\item $P(\varphi_i) \subseteq \{q,r\}$. In this case, we can apply the same argument as before, by simply switching the roles of $p$ and $q$.

\end{compactitem}

\medskip

\noindent Therefore, $\varphi$ cannot exist, and this means that $\psi$ cannot be expressed in \KnKrom within the same propositional alphabet.
\end{proof}

\begin{corollary}\label{th:krom_vs_horn}
The following results hold:
\medskip
\begin{compactenum}
\item \KnHorn and \KnKrom are $\preceq^w$-incomparable;
\item $\KnCore\prec^w \KnKrom,\KnHorn$.
\end{compactenum}
\end{corollary}

\begin{proof}
As we have seen in Theorem~\ref{th:horn_vs_bool}, the \KnKrom-formula $p \lor q$ cannot be translated into \KnHorn within the same propositional alphabet, and, as we have seen in Theorem~\ref{th:krom_vs_bool}, the \KnHorn-formula $p \land q \rightarrow r$ cannot be translated into \KnKrom under the same conditions. These two observations, together, prove that we cannot compare \KnHorn and \KnKrom, under the weak notion of expressivity. As an immediate consequence, since $\KnCore = \KnHorn \cap \KnKrom$, we have that $\KnCore\prec^w \KnHorn$ and $\KnCore\prec^w\KnKrom$.
\end{proof}

\section{Box and Diamond Fragments}\label{sec:kn_expr2}

In this section, we study the relative expressive power for box and diamond fragments, starting with sub-Horn fragments without diamonds. First of all, we prove the following useful property of the fragments \KnHornBox and \KnCoreBox. Consider two models $M_1,M_2$ such that all $M_i=(\mathcal F,V_i)$ are based on the same relational frame. We define the {\em intersection} model as the unique model $M_{M_1\cap M_2}=(\mathcal F,V_{V_1\cap V_2})$, where, for each $w\in W$, $V_{V_1\cap V_2}(w)=V_1(w)\cap V_2(w).$

\begin{lemma}\label{lem:intersection}
\KnHornBox\ is closed under intersection of models.
\end{lemma}

\begin{proof}
Let $\varphi=\varphi_1\wedge\ldots\wedge\varphi_l$ a \KnHornBox-formula such that $M_1,w\mmodels\varphi$ and $M_2,w\mmodels\varphi$, where $M_1=(\mathcal F,V_1)$ and $M_2=(\mathcal F,V_2)$; we want to prove that $M_{M_1\cap M_2},w\mmodels\varphi$. Suppose, by way of contradiction, that $M_{M_1\cap M_2},w\not\mmodels\varphi$. Then, there must be some $\varphi_i$ such that $M_{M_1\cap M_2},w\not\mmodels\varphi_i$. As in Theorem~\ref{th:horn_vs_bool}, we can assume that $\varphi_i$ is a clause of the type $\nabla(\lambda_1\wedge\ldots\wedge\lambda_n\rightarrow\lambda)$. This means that $M_{M_1\cap M_2},w'\mmodels\lambda_1\wedge\ldots\wedge\lambda_n$ and $M_{M_1\cap M_2},w'\not\mmodels\lambda$ for some $w'$. We want to prove that, for each $1\le j\le n$, both $M_1$ and $M_2$ satisfy $\lambda_j$ at $w'$. To see this, we reason by induction on $md(\lambda_j)$. If $md(\lambda_j)=0$, then $\lambda_j=p$ for some propositional letter $p$; but if $M_{M_1\cap M_2},w'\mmodels p$, then $p\in V_1(w')\cap V_2(w')$, which means that  $M_1,w'\mmodels p$ and  $M_2,w'\mmodels p$. If $md(\lambda_j)>0$, then $\lambda_j=\boxx\lambda'$. Since $M_{M_1\cap M_2},w'\mmodels \boxx\lambda'$, for every $v$ such that $w'\Rx v$ it is the case that $M_{M_1\cap M_2},v\mmodels \lambda'$. Thus, for every $v$ such that $w'\Rx v$, we know by inductive hypothesis that $M_1,v\mmodels\lambda'$ and $M_2,v\mmodels\lambda'$. But this immediately implies that $M_1,w'\mmodels\boxx\lambda'$ and $M_2,v\mmodels\boxx\lambda'$, which completes the induction. Now, we know that  $M_1,w'\mmodels\lambda_1\wedge\ldots\wedge\lambda_n$ and
$M_2,w'\mmodels\lambda_1\wedge\ldots\wedge\lambda_n$; therefore, $M_1,w'\mmodels\lambda$ and $M_2,w'\mmodels\lambda$. A similar inductive argument
shows that $M_{M_1\cap M_2},w'\mmodels\lambda$, implying that $M_{M_1\cap M_2},w\mmodels\varphi_i$; but this contradicts  our hypothesis that $M_{M_1\cap M_2},w\not\mmodels\varphi$.
\end{proof}

\begin{theorem}\label{th:hornbox_vs_boxdia}
The following relationships hold:
\medskip
\begin{compactenum}
\item $\KnHornBox\prec\KnHorn$;
\item $\KnCoreBox\prec\KnCore$.
\end{compactenum}
\end{theorem}

\begin{proof}
Since $\KnHornBox$ (resp., \KnCoreBox) is a syntactical fragment of $\KnHorn$ (resp., \KnCore), we know that $\KnHornBox\preceq\KnHorn$ and $\KnCoreBox\preceq\KnCore$. It remains to be proved that there exists a formula that belongs to $\KnHorn$ (resp., \KnCore) and that cannot be translated to $\KnHornBox$ (resp., $\KnCoreBox$) over any finite extension of the propositional alphabet. Here, we prove that this is the case for a \KnCore-formula (which is a \KnHorn-formula as well) that cannot be translated to  \KnHornBox (and, therefore, to \KnCoreBox, either). Let $\mathcal P=\{p\}$, consider the \KnHorn-formula
$$\psi = \diax p,$$

\noindent and suppose by contradiction that there exists a propositional alphabet $\mathcal P'\supseteq \mathcal P$ and a \KnHornBox formula $\varphi$ written over $\mathcal P'$ such that for every model $M$ over the propositional alphabet $\mathcal P$ and every world $w$ we have that $M,w\mmodels \psi$ if and only if there exists $M^{\mathcal P'}$ such that $M^{\mathcal P'},w\mmodels\varphi.$ Let $M_1=(\mathcal F,V_1)$ and $M_2=(\mathcal F,V_2)$, where $\mathcal F$ is based on the set $W=\{w_0,w_1,w_2\}$. Let $w_0\Rx w_1$ and $w_0\Rx w_2$, and define the valuation functions $V_1,V_2$ as follows:

$$
V_i(w_j)=\left\{\begin{array}{ll}
\{p\} & \mbox{if }i=j,\\
\emptyset & \mbox{otherwise}.
\end{array}\right.
$$

\noindent Clearly, $M_1,w_0\mmodels\psi$ and $M_2,w_0\mmodels\psi$; since $\varphi$ is a \KnHornBox-translation of $\psi$, it must be the case that, for some extensions $M_1^{\mathcal P'}$ and $M_2^{\mathcal P'}$, we have that $M_1^{\mathcal P'},w_0\mmodels\varphi$ and $M_2^{\mathcal P'},w_0\mmodels\varphi$. By Lemma~\ref{lem:intersection}, their intersection model $M_{M_1^{\mathcal P'}\cap M_2^{\mathcal P'}}$ is such that $M_{M_1^{\mathcal P'}\cap M_2^{\mathcal P'}},w_0\mmodels\varphi$. But $p \not\in V_{V_1^{\mathcal P'}\cap V_2^{\mathcal P'}}(w)$ for every $w\in W$, so $M_{M_1^{\mathcal P'}\cap M_2^{\mathcal P'}},w\not\mmodels\psi$. This contradicts the hypothesis that $\varphi$ is a translation of $\psi$.
\end{proof}

To establish the expressive power of \KnHornDia and \KnCoreDia with respect to other fragments, we now prove a closure property similar to Lemma~\ref{lem:intersection}. Consider two models $M_1 = (\mathcal F_1, V_1)$, $M_2 = (\mathcal F_2, V_2)$ based on two (possibly different) relational frames $\mathcal F_1=(W_1,\{R_1\}_{\alpha\in\tau})$ and $\mathcal F_2=(W_2,\{R_2\}_{\alpha\in\tau})$. We define the {\em product} model as the unique model $M_{M_1\times M_2}=(\mathcal F_{\mathcal F_1 \times \mathcal F_2},V_{V_1\times V_2})$, where: 
\begin{inparaenum}[\it (i)]
\item $\mathcal F_{\mathcal F_1 \times \mathcal F_2} = (W_1 \times W_2, \{R_{R_1\times R_2}\}_{\alpha \in \tau})$, that is, worlds are all and only the pairs of worlds from $W_1$ and $W_2$; 
\item for every $\alpha \in \tau$, $(w_1,w_2) R_{R_1\times R_2,\alpha} (w_1',w_2')$ if and only if $w_1 R_{1,\alpha} w_1'$ and $w_2 R_{2,\alpha} w_2'$, that is, worlds in  $\mathcal F_{\mathcal F_1 \times \mathcal F_2}$ are connected to each other as the component worlds were connected in $\mathcal F_1$ and $\mathcal F_2$; and 
\item $V_{V_1\times V_2}((w_1,w_2)) = V_1(w_1) \cap V_2(w_2)$. 
\end{inparaenum}

\begin{lemma}\label{lem:product}
\KnHornDia\ is closed under product of models.
\end{lemma}

\begin{proof}
Let $\varphi=\varphi_1\wedge\ldots\wedge\varphi_l$ be a \KnHornDia-formula such that $M_1,w_1\mmodels\varphi$ and $M_2,w_2\mmodels\varphi$. We want to prove that $M_{M_1\times M_2},(w_1,w_2)\mmodels\varphi$; suppose by way of contradiction, that $M_{M_1\times M_2},(w_1,w_2)\not\mmodels\varphi$. Then, there must be some $\varphi_i$ such that $M_{M_1\times M_2},(w_1,w_2)\not\mmodels\varphi_i$. As in Theorem~\ref{th:horn_vs_bool}, we can assume that $\varphi_i$ is a clause of the type $\nabla(\lambda_1\wedge\ldots\wedge\lambda_n\rightarrow\lambda)$. This means that $M_{M_1\times M_2},(w_1',w_2')\mmodels\lambda_1\wedge\ldots\wedge\lambda_n$ and $M_{M_1\times M_2},(w_1',w_2')\not\mmodels\lambda$ for some $(w_1',w_2')$. We want to prove that, for each $1\le j\le n$,  $M_1$ and $M_2$ satisfy $\lambda_j$ at, respectively, $w_1'$ and $w_2'$. To see this, we reason by induction on  $md(\lambda_j)$. If $md(\lambda_j)=0$, then $\lambda_j=p$ for some propositional letter $p$: by the definition of product, we have that $M_{M_1\times M_2},(w_1',w_2')\mmodels p$ iff $p\in V_1(w_1')\cap V_2(w_2')$, which means that  $M_1,w_1'\mmodels p$ and  $M_2,w_1'\mmodels p$. If $md(\lambda_j)>0$, then $\lambda_j=\diax\lambda'$. Since $M_{M_1\times M_2},(w_1',w_2')\mmodels \diax\lambda'$, then there exists $(v_1,v_2)$ such that $(w_1',w_2')R_{R_1\times R_2,\alpha} (v_1,v_2)$ and $M_{M_1\times M_2},(v_1,v_2)\mmodels \lambda'$.
We know by inductive hypothesis that $M_1,v_1\mmodels\lambda'$ and $M_2,v_2\mmodels\lambda'$ and that, by definition of product, $w_1'R_{1,\alpha} v_1$ and $w_2'R_{2,\alpha} v_2$. But this immediately implies that $M_1,w_1'\mmodels\diax\lambda'$ and $M_2,w_2'\mmodels\diax\lambda'$, which completes the induction. Now, we know that  $M_1,w_1'\mmodels\lambda_1\wedge\ldots\wedge\lambda_n$ and
$M_2,w_2'\mmodels\lambda_1\wedge\ldots\wedge\lambda_n$; therefore, $M_1,w_1'\mmodels\lambda$ and $M_2,w_2'\mmodels\lambda$. A similar inductive argument
shows that $M_{M_1\times M_2},(w_1',w_2')\mmodels\lambda$, implying that $M_{M_1\times M_2},(w_1,w_2)\mmodels\varphi_i$, in contradiction with  the hypothesis that $M_{M_1\times M_2},(w_1,w_2)\not\mmodels\varphi$.
\end{proof}

\begin{theorem}\label{th:horndia_vs_boxdia}
The following relationships hold:
\medskip
\begin{compactenum}
\item $\KnHornDia\prec\KnHorn$;
\item $\KnCoreDia\prec\KnCore$.
\end{compactenum}
\end{theorem}

\begin{proof}
Since $\KnHornDia$ (resp., \KnCoreDia) is a syntactical fragment of $\KnHorn$ (resp., \KnCore), we know that $\KnHornDia\preceq\KnHorn$ and $\KnCoreDia\preceq\KnCore$. It remains to be proved that there exists a formula that belongs to $\KnHorn$ (resp., \KnCore) and that cannot be translated to $\KnHornDia$ (resp., $\KnCoreDia$) over any finite extension of the propositional alphabet. Here, we prove that this is the case for a \KnCore-formula (which is a \KnHorn-formula as well) that cannot be translated to  \KnHornDia (and, therefore, to \KnCoreDia, either). Let $\mathcal P=\{p, q\}$, consider the \KnHorn-formula
$$\psi = \boxx p \rightarrow q,$$

\noindent and suppose by contradiction that there exists a propositional alphabet $\mathcal P'\supseteq \mathcal P$ and a \KnHornDia formula $\varphi$ written over $\mathcal P'$ such that for every model $M$ over the propositional alphabet $\mathcal P$ and every world $w$ we have that $M,w\mmodels \psi$ if and only if there exists $M^{\mathcal P'}$ such that $M^{\mathcal P'},w\mmodels\varphi.$ Let $M_1=(\mathcal F_1,V_1)$ and $M_2=(\mathcal F_2,V_2)$, where $\mathcal F_1$ is based on the set $W=\{w_0,w_1\}$ and such that $w_0\Rx w_1$, while $\mathcal F_2$ is based on $\{v_0\}$ and such that $\Rx = \emptyset$. Define the valuation function $V_1$ as always empty, and let $q\in V_2(v_0)$. Clearly, $M_1,w_0\mmodels\psi$ and $M_2,v_0\mmodels\psi$. Since $\varphi$ is a \KnHornDia-translation of $\psi$, it must be the case that, for some extensions $M_1^{\mathcal P'}$ and $M_2^{\mathcal P'}$, we have that $M_1^{\mathcal P'},w_0\mmodels\varphi$ and $M_2^{\mathcal P'},v_0\mmodels\varphi$. By Lemma~\ref{lem:product}, their product model $M_{M_1^{\mathcal P'}\times M_2^{\mathcal P'}}$ is such that $M_{M_1^{\mathcal P'}\times M_2^{\mathcal P'}},(w_0,v_0)\mmodels\varphi$. Notice that $q \not\in V_{V_1^{\mathcal P'}\times V_2^{\mathcal P'}}(w_0,v_0)$ and that $(w_0,v_0)$ has no $\Rx$-successors. Hence, we have that $M_{M_1^{\mathcal P'}\times M_2^{\mathcal P'}},(w_0,v_0)\mmodels\boxx p$ but $M_{M_1^{\mathcal P'}\times M_2^{\mathcal P'}},(w_0,v_0)\not\mmodels q$, in contradiction with the hypothesis that $\varphi$ is a translation of $\psi$. Therefore, $\varphi$ cannot exist, and this means that $\psi$ cannot be expressed in \KnHornDia within any finite extension of the propositional alphabet.
\end{proof}

The argument of Theorem~\ref{th:hornbox_vs_boxdia}, based on the intersection of models, cannot be replicated to establish the relationship between $\KnKromBox$ and $\KnKrom$. It turns out that in this case the possibility of expanding the propositional alphabet does make the difference, as the following result shows.

\begin{theorem}\label{th:krombox_vs_krom}
The following relationships hold:
\medskip
\begin{compactenum}
\item $\KnKromBox\equiv\KnKrom$;
\item $\KnKromBox\prec^w\KnKrom$.
\end{compactenum}
\end{theorem}

\begin{proof}
The first result is easy to prove. Suppose that
$$\varphi=\nabla_1(\lambda_1^1\vee\lambda_2^1)\wedge \nabla_2(\lambda_1^2\vee\lambda_2^2)\wedge\ldots\wedge\nabla_i(\lambda_1^i\vee\lambda_2^i)\wedge\ldots
\wedge\nabla_l(\lambda_1^l\vee\lambda_2^l)$$

\noindent is a $\KnKrom$-formula, where, as always, we treat literals as special clauses. There are two cases. First, suppose that $\lambda_1^i=\diax\lambda$, for some $1\leq i\leq l$, where $\lambda$ is a positive literal. We claim that the \KnKromBox-formula
$$\varphi'=\nabla_1(\lambda_1^1\vee\lambda_2^1)\wedge \nabla_2(\lambda_1^2\vee\lambda_2^2)\wedge\ldots\wedge\nabla_i(\neg\boxx p\vee\lambda_2^i)\wedge\nabla_i\boxx(p\vee\lambda)\wedge\ldots
\wedge\nabla_l(\lambda_1^l\vee\lambda_2^l),$$

\noindent where $p$ is a fresh propositional variable, is equi-satisfiable to $\varphi$. To see this, let $\mathcal P$ the propositional alphabet in which $\varphi$ is written, and let $\mathcal P'=\mathcal P\cup\{p\}$, and consider a model $M=(\mathcal F,V)$ such that, for some world $w$, it is the case that $M,w\mmodels\varphi$; in particular, it is the case that $M,w\mmodels\nabla_i(\lambda_1^i\vee\lambda_2^i)$; let $W_i\subseteq W$ be the set of worlds reachable from $w$ via the universal prefix $\nabla_i$, and consider $v\in W_i$. If $M,v\mmodels\lambda_2^i$ we can extend $M$ to a model $M^\mathcal P=(\mathcal F,V^\mathcal P)$ such that it satisfies $p$ on every world $\alpha$-reachable from $v$, if any, and both substituting clauses are satisfied. If, on the other hand, $M,v\mmodels\diax\lambda$, for some $t$ such that $v\Rx t$ we have that $M,t\mmodels\lambda$; we can now extend $M$ to a model $M^\mathcal P=(\mathcal F,V^\mathcal P)$ such that it satisfies $\neg p$ on $t$, and $p$ on every other world reachable from $v$, if any, and, again, both substituting clauses are satisfied. A reversed argument proves that if $M,w\mmodels\varphi'$ it must be the case that $M,w\mmodels\varphi$. If, as a second case, $\lambda_1^i=\neg\diax\lambda$, where $\lambda$ is a positive literal, then the translating formula is
$$\varphi'=\nabla_1(\lambda_1^1\vee\lambda_2^1)\wedge \nabla_2(\lambda_1^2\vee\lambda_2^2)\wedge\ldots\wedge\nabla_i(\boxx p\vee\lambda_2^i)\wedge\nabla_i\boxx(\neg p\vee\neg \lambda)\wedge\ldots
\wedge\nabla_l(\lambda_1^l\vee\lambda_2^l),$$

\noindent and the proof of equi-satisfiability is identical to the above one.

\medskip

In order to prove the second result, we observe that since $\KnKromBox$ is a syntactical fragment of $\KnKrom$ we know that $\KnKromBox\preceq^w\KnKrom$. It remains to be proved that there exists a formula that belongs to $\KnKrom$ and that cannot be translated to $\KnKromBox$ within the same propositional alphabet. Let $\mathcal P=\{p\}$, consider the \KnHorn-formula
$$\psi = \diax p,$$

\noindent and suppose by contradiction that there exists a \KnKromBox formula $\varphi$ such that for every model $M$ over the propositional alphabet $\mathcal P$ and every world $w$ we have that $M,w\mmodels \psi$ if and only if $M,w\mmodels\varphi.$ Once again, we can safely assume that $\varphi=\varphi_1\wedge\varphi_2\wedge\ldots\wedge\varphi_l$, and that each $\varphi_i$ is a clause. Consider a model $M = \langle \mathcal F, V\rangle$, where $\mathcal F$ is based on the set of worlds $W$, and let $w\in W$ be a world such that $M,w\not\mmodels\psi$. Such a model must exist since $\psi$ is not a tautology. Since $\varphi$ is a conjunction of Krom clauses, we have that there must exist at least one clause $\varphi_i = \nabla(\lambda_1 \vee \lambda_2)$ such that $M, w \not\mmodels \varphi_i$. Hence, there must exist a world $w'$ such that $M, w' \not\mmodels (\lambda_1 \vee \lambda_2)$. Now, consider the model $M^*$ obtained from $M$ by extending the set of worlds $W$ to $W^*=W\cup\{w^*\}$, in such a way that $w\Rx^* w^*$ and that $p\in V^*(w^*)$; clearly, $M^*,w\mmodels\psi$. We want to prove that $M^*,w'\not\mmodels\lambda_1\vee\lambda_2$. Let us prove the following:
$$M,t\mmodels\lambda\Leftrightarrow M^*,t\mmodels\lambda,$$

\noindent for every $t\in W$ and positive literal $\lambda$. We do so by induction on $md(\lambda)$. If $md(\lambda)=0$, then $\lambda$ is a propositional letter (the cases in which $\lambda=\top$ is trivial): the valuation of $t$ has not changed from $M$ to $M^*$, and therefore we have the claim immediately. If $md(\lambda)>0$, then we have two cases:

\medskip

\begin{itemize}
\item $\lambda=\Box_\beta\lambda'$, and $\beta\neq\alpha$. In this case the claim holds trivially, as the $\beta$-structure has not changed from $M$ to $M^*$.

\item $\lambda=\boxx\lambda'$, and $\lambda'$ is a positive literal. By definition, $M,t\mmodels\boxx\lambda'$ if and only if for every $t'$ such that $t\Rx t'$, if any, it is the case that $M,t'\mmodels\lambda'$. Clearly, if $t\neq w$, the set of reachable worlds from $t$ has not changed, and thanks to the inductive hypothesis, $M,t'\mmodels \lambda'$ if and only if $M^*,t'\mmodels \lambda'$; therefore, $M,t'\mmodels\lambda$ if and only if $M^*,t'\mmodels \lambda$ as we wanted. Otherwise, suppose that $t=w$. If $M,t\not\mmodels\boxx\lambda'$, then: \begin{inparaenum}[\it (i)] \item $\lambda'\neq\top$, because $\boxx\top$ is always satisfied, and \item there exist some $t'$ such that $t\Rx t'$ and $M,t'\not\mmodels\lambda'$, and $t'\neq w^*$ (since $w^*$ is a new world); so, by inductive hypothesis, $M^*,t'\not\mmodels\lambda'$, which means that $M^*,t\not\mmodels\boxx\lambda'$. \end{inparaenum} If, on the other hand, $M,t\mmodels\boxx\lambda'$, then: \begin{inparaenum}[\it (i)] \item if $\lambda'=\top$, then $M^*,t\mmodels\boxx\top$ independently from the presence of $w^*$; \item if $\lambda'=\Box_\beta\lambda''$ for some relation $\beta$, then observe that $M^*,w^*\mmodels\Box_\beta\lambda''$ because $w^*$  has no $\beta$-successors for any relation $\beta$, and, hence, $M^*,t\mmodels\boxx\lambda'$, and \item if $\lambda'=p$, then $M^*,t\mmodels\boxx\lambda'$ because $w\Rx^* w^*$ and $p\in V^*(w^*)$ by construction. \end{inparaenum} 
\end{itemize}

\medskip

\noindent Since by hypothesis $M,w'\not\mmodels\lambda_1\vee\lambda_2$, the above argument implies that $M^*,w'\not\mmodels\lambda_1\vee\lambda_2$, which means that $M^*,w\not\mmodels\varphi_i$, that is, $M^*,w\not\mmodels\varphi$. Therefore $\varphi$ cannot be a translation of $\psi$, and the claim is proved.
\end{proof}

The following result deals with sub-Krom fragments without boxes; as before, the argument of Theorem~\ref{th:horndia_vs_boxdia}, based on the product of models, cannot be replicated.

\begin{theorem}\label{th:kromdia_vs_boxdia}
The following relationships hold:
\medskip
\begin{compactenum}
\item $\KnKromDia\equiv\KnKrom$;
\item $\KnKromDia\prec^w\KnKrom$;
\end{compactenum}
\end{theorem}

\begin{proof}
The first result is relatively easy to see. Suppose that
$$\varphi=\nabla_1(\lambda_1^1\vee\lambda_2^1)\wedge \nabla_2(\lambda_1^2\vee\lambda_2^2)\wedge\ldots\wedge\nabla_i(\lambda_1^i\vee\lambda_2^i)\wedge\ldots
\wedge\nabla_l(\lambda_1^l\vee\lambda_2^l)$$

\noindent is a $\KnKrom$-formula, where, as always, we treat literals as special clauses. There are two cases. Suppose, first, that $\lambda_1^i=\boxx\lambda$, where $\lambda$ is a positive literal. We claim that the \KnKromDia-formula
$$\varphi'=\nabla_1(\lambda_1^1\vee\lambda_2^1)\wedge \nabla_2(\lambda_1^2\vee\lambda_2^2)\wedge\ldots\wedge\nabla_i(\neg\diax p\vee\lambda_2^i)\wedge\nabla_i\boxx(p\vee\lambda)\wedge\ldots
\wedge\nabla_l(\lambda_1^l\vee\lambda_2^l),$$

\noindent where $p$ is a fresh propositional variable, is equi-satisfiable to $\varphi$. To see this, let $\mathcal P$ the propositional alphabet in which $\varphi$ is written, and let $\mathcal P'=\mathcal P\cup\{p\}$, and consider a model $M=(\mathcal F,V)$ such that, for some world $w$, it is the case that $M,w\mmodels\varphi$; in particular, it is the case that $M,w\mmodels\nabla_i(\lambda^i_1\vee \lambda^i_2)$; let $W_i\subseteq W$ be the set of worlds reachable from $w$ via the universal prefix $\nabla_i$, and consider $v\in W_i$. If $M,v\mmodels\lambda_2^i$ we can extend $M$ to a model $M^\mathcal P=(\mathcal F,V^\mathcal P)$ such that it satisfies $p$ on every world $\alpha$-reachable from $v$, if any, and both substituting clauses are satisfied. If, on the other hand, $M,v\mmodels\boxx\lambda_1^i$, for every $t$ such that $v\Rx t$ we have that $M,t\mmodels\lambda$; we can now extend $M$ to a model $M^\mathcal P=(\mathcal F,V^\mathcal P)$ such that it satisfies $\neg p$ on every such $t$ (if any), and, again, both substituting clauses are satisfied. A reversed argument proves that if $M,w\mmodels\varphi'$ it must be the case that $M,w\mmodels\varphi$. If, as a second case, $\lambda_1^i=\neg\boxx\lambda$, where $\lambda$ is a positive literal, then the translating formula is
$$\varphi'=\nabla_1(\lambda_1^1\vee\lambda_2^1)\wedge \nabla_2(\lambda_1^2\vee\lambda_2^2)\wedge\ldots\wedge\nabla_i(\diax p\vee\lambda_2^i)\wedge\nabla_i\boxx(\neg p\vee\neg \lambda)\wedge\ldots
\wedge\nabla_l(\lambda_1^l\vee\lambda_2^l),$$

\noindent and the proof of equi-satisfiability is identical to the above one.

\medskip

As for the second relationship, since $\KnKromDia$ is a syntactical fragment of $\KnKrom$, we know that $\KnKromDia\preceq\KnKrom$. It remains to show that the relationship is strict. To this end, we consider the following $\KnKrom$-formula and we prove that it cannot be translated to $\KnKromDia$ within the same propositional alphabet:
$$\psi=\boxx p\rightarrow q.$$

Suppose, by contradiction, that there exist a conjunction $\varphi$ of box-free Krom clauses, such that for every model $M$ over the propositional alphabet $\mathcal P=\{p,q\}$ and every world $w$ we have that $M,w\mmodels \psi$ if and only if $M,w\mmodels\varphi.$ Let $\varphi=\varphi_1 \land \ldots \land \varphi_n$, where each $\varphi_i$ is in its generic form $\nabla(\lambda_1 \vee \lambda_2)$, with $\lambda_1$ and $\lambda_2$ either positive or negative literals. As always, literals are treated as special clauses. Now, consider a model $M = \langle \mathcal F, V\rangle$, where $\mathcal F$ is based on the set of worlds $W$, and let $w\in W$ be a world such that $M,w\not\mmodels\psi$, and that exists at least one $v$ such that $w\Rx v$. Since $M,w\not\mmodels\psi$,  we have that $q\notin V(w)$ and for each $v$ such that $w\Rx v$ it is the case that $p\in V(v)$. Since $\varphi$ is a translation of $\psi$, it must be the case that $M,w\not\mmodels\varphi$, which implies that there must be a clause $\varphi_i$ such that $M,w\not\mmodels\varphi_i$, that is, there must be a world $w'$ such that  $M,w'\not\mmodels(\lambda_1\vee\lambda_2)$. Now, consider the model $M^*$ obtained from $M$ by extending the set of worlds $W$ to $W^*=W\cup\{w^*\}$, in such a way that $w\Rx^\ast w^*$  and that $V^*(w^*)=\emptyset$; clearly, $M^*,w\mmodels \psi$. We want to prove that $M^*,w'\not\mmodels\varphi_i$. Let us prove the following:
$$M,t\mmodels\lambda\Leftrightarrow M^*,t\mmodels\lambda,$$

\noindent for every $t\in W$ and positive literal $\lambda$. We do so by induction on $md(\lambda)$. If $md(\lambda)=0$, then $\lambda$ is a propositional letter (the cases in which $\lambda=\top$ are trivial): the valuation of $t$ has not changed from $M$ to $M^*$, and therefore we have the claim immediately. If $md(\lambda)>0$, then there are two cases:

\begin{itemize}
\item $\lambda=\Diamond_\beta\lambda'$, and $\beta\neq\alpha$. In this case the claim holds trivially, as the $\beta$-structure has not changed from $M$ to $M^*$.

\item $\lambda=\diax\lambda'$, and $\lambda'$ is a positive literal. By definition, $M,t\mmodels\diax\lambda'$ if and only if there exist some $t'$ such that $t\Rx t'$ and $M,t'\mmodels\lambda'$. Clearly, if $t\neq w$, the set of reachable worlds from $t$ has not changed, and thanks to the inductive hypothesis, $M,t'\mmodels \lambda'$ if and only if $M^*,t'\mmodels \lambda'$; therefore, $M,t\mmodels\diax\lambda$ if and only if $M^*,t\mmodels\diax\lambda$, as we wanted. Otherwise, suppose that $t=w$. If $M,t\mmodels\diax\lambda'$, then there exist some $t'$ such that $t\Rx t'$ and $M,t'\mmodels\lambda'$, and $t'\neq w^*$ (since $w^*$ is a new world); so, by inductive hypothesis, $M,t'\mmodels\lambda'$, which means that $M^*,t\mmodels\diax\lambda'$. If, on the other hand, $M,t\not\mmodels\diax\lambda'$, then: \begin{inparaenum}[\it (i)] \item $\lambda'\neq\top$, because we have built $M$ in such a way that $w$ has a $\alpha$-successor, and \item for every $t'$ such that $t\Rx t'$ it is the case that $M,t'\not\mmodels\lambda'$. \end{inparaenum} Since $V(w^*)=\emptyset$, and $\lambda'$ is positive, for every $t'$ such that $t\Rx^\ast t'$ it is the case that $M^*,t'\not\mmodels\lambda'$, and, therefore, $M^*,t\not\mmodels\diax\lambda'$, as we wanted. 

\end{itemize}

\noindent This means that $M,w'\not\models\lambda_1\vee\lambda_{2}$ implies that  $M^*,w'\not\models\lambda_1\vee\lambda_{2}$, that is, $M^*,w'\not\models\varphi_i$. This implies that $M^*,w\not\mmodels\varphi$. Therefore, $\varphi$ cannot exist, and this means that $\psi$ cannot be expressed in \KnKromDia within the same propositional alphabet.
\end{proof}

\begin{corollary}
The following results hold:
\medskip
\begin{compactenum}
\item \KnHornBox and \KnHornDia are $\prec$-incomparable;
\item \KnKromBox and \KnKromDia are $\prec^w$-incomparable;
\item \KnCoreBox and \KnCoreDia are $\prec$-incomparable.
\end{compactenum}
\end{corollary}

\begin{proof}
As we have seen in Theorem~\ref{th:hornbox_vs_boxdia}, the \KnCoreDia-formula (which is also a \KnHornDia-formula) $\diax p$ cannot be translated into \KnHornBox (and therefore it cannot be translated to \KnCoreBox either), over any finite extension of the propositional alphabet, and, as we have seen in Theorem~\ref{th:horndia_vs_boxdia}, the  \KnCoreBox-formula $\boxx p\rightarrow q$ (which is also a \KnHornBox-formula) cannot be translated into \KnHornDia (and therefore it cannot be translated to \KnCoreDia either), over any finite extension of the propositional alphabet. These two observations, together, show that we cannot compare  \KnHornBox with \KnHornDia, nor \KnCoreBox with \KnCoreDia. Similarly, Theorem~\ref{th:krombox_vs_krom} proves that the \KnKromDia-formula $\diax p$ cannot be translated to \KnKromBox, and Theorem~\ref{th:kromdia_vs_boxdia} proves that the \KnKromBox-formula $\boxx p\rightarrow q$ cannot be translated to \KnKromDia, all this within the same propositional alphabet; these two observations, together, imply that, at least within the same propositional alphabet, we cannot compare \KnKromBox and \KnKromDia, either.
\end{proof}

\begin{corollary}
The following results hold:

\medskip

\begin{compactenum}
\item $\KnHorn^\clubsuit,\KnHorn$ cannot be $\preceq^w$-compared with $\KnKrom^\spadesuit,\KnKrom$, and viceversa, for $\clubsuit,\spadesuit\in\{\Box,\lozenge\}$;
\item $\KnCoreBox\prec^w\KnHornBox$ and $\KnCoreDia\prec^w\KnHornDia$;
\item $\KnCoreBox,\KnCoreDia\prec\KnKrom$, $\KnKromBox$, $\KnKromDia$.
\end{compactenum}
\end{corollary}

\begin{proof}
As far as the first result is concerned, as we have seen in Theorem~\ref{th:horn_vs_bool}, the formula $p\vee q$, which belongs to all sub-Krom fragments of \KnBool, cannot be translated to \KnHorn, and, therefore, it cannot be translated to any sub-Horn fragment either, at least within the same propositional alphabet.  Theorem~\ref{th:krom_vs_bool}, on the other hand, proves that the formula $(p\wedge q)\rightarrow r$, which belongs to all sub-Horn fragments of \KnBool, cannot be translated to \KnKrom, and, therefore, it cannot be translated to any sub-Krom fragment either, at least within the same propositional alphabet. These two observations, together, imply that the claim holds. Thanks to the above result, an taking into account that $\KnCoreBox = \KnHornBox \cap \KnKromBox$ and that $\KnCoreDia = \KnHornDia \cap \KnKromDia$, the second claim immediately follow. Finally, to prove the third result it is sufficient to recall that the proof of Theorem~\ref{th:hornbox_vs_boxdia} shows that the \KnKrom-formula $\diax p$ cannot be translated into \KnCoreBox (over any finite extension of the propositional alphabet), while the proof of Theorem~\ref{th:horndia_vs_boxdia} shows that the \KnKrom-formula $\boxx p \rightarrow q$ cannot be translated into \KnCoreDia (over any finite extension of the propositional alphabet). Thanks to Theorem~\ref{th:krombox_vs_krom} and Theorem~\ref{th:kromdia_vs_boxdia} we know that $\KnKrom \equiv \KnKromBox \equiv \KnKromDia$ and we have the claim.
\end{proof}


\section{Conclusions}\label{sec:concl}

In this paper we studied the relative expressive power of several sub-propositional fragments of the multi-modal logic \Kn. Inspired by recent work on sub-propositional fragments of temporal and description logic~\cite{DBLP:conf/lpar/ArtaleKRZ13,DBLP:journals/tocl/ArtaleKRZ14,artale2015,DBLP:conf/jelia/BresolinMS14}, we defined the Horn and the Krom fragments of modal logic, and their box and diamond fragments. We compared the relative expressive power of the fragments at two different levels, characterized by respectively allowing or not allowing new propositional letters in the translations, and the results are shown in Fig.~\ref{fig:exprpower}. In most cases relative expressivity coincides with syntactical containment, with the notable exception of the Krom fragments, that are expressively equivalent, but not weakly expressively equivalent. Because of our very general approach for comparing the expressive power of languages, most of our result can be transferred to other sub-propositional modal logic such as the fragments of \LTL without Since and Until studied in~\cite{DBLP:conf/lpar/ArtaleKRZ13} and the sub-propositional fragments of \HS~\cite{artale2015,bresolin:2016,DBLP:conf/jelia/BresolinMS14}. To the best of our knowledge, this is the first work where sub-Krom and sub-Horn fragments of \Kn have been considered. 

As future work, it would be desirable to complete this picture relatively to the strong hierarchy (although extending the current results do not seem trivial), and to study the complexity of the fragments that are expressively weaker or incomparable to \KnHorn. Because of their lower expressive power, the satisfiability problem for sub-Krom and sub-Horn fragments may have a lower complexity than full \Kn, as our preliminary results seem to suggest, not only in the case of \Kn, but, also, for some of its most common axiomatic extensions. 

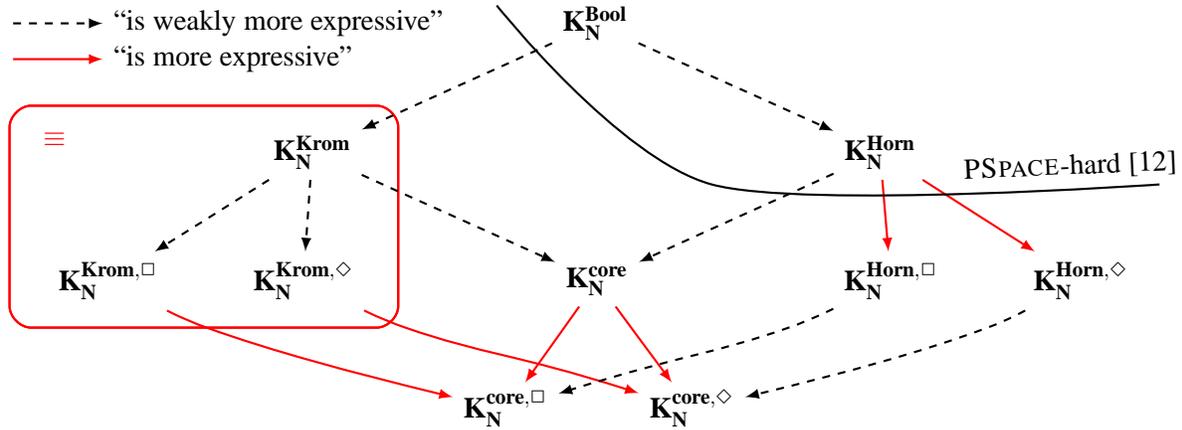
\begin{figure}[t!]
   \centering
\begin{tikzpicture}[>=latex,line join=bevel,scale=0.675,thick]
\begin{scope}
  \pgfsetstrokecolor{black}
  \definecolor{strokecol}{rgb}{1.0,0.0,0.0};
  \pgfsetstrokecolor{strokecol}
  \draw (20.0bp,80.0bp) .. controls (20.0bp,80.0bp) and (213.0bp,80.0bp)  .. (213.0bp,80.0bp) .. controls (219.0bp,80.0bp) and (225.0bp,86.0bp)  .. (225.0bp,92.0bp) .. controls (225.0bp,92.0bp) and (225.0bp,192.0bp)  .. (225.0bp,192.0bp) .. controls (225.0bp,198.0bp) and (219.0bp,204.0bp)  .. (213.0bp,204.0bp) .. controls (213.0bp,204.0bp) and (20.0bp,204.0bp)  .. (20.0bp,204.0bp) .. controls (14.0bp,204.0bp) and (8.0bp,198.0bp)  .. (8.0bp,192.0bp) .. controls (8.0bp,192.0bp) and (8.0bp,92.0bp)  .. (8.0bp,92.0bp) .. controls (8.0bp,86.0bp) and (14.0bp,80.0bp)  .. (20.0bp,80.0bp);
\end{scope}
\begin{scope}
  \pgfsetstrokecolor{black}
  \definecolor{strokecol}{rgb}{1.0,0.0,0.0};
  \pgfsetstrokecolor{strokecol}
  \draw (20.0bp,80.0bp) .. controls (20.0bp,80.0bp) and (213.0bp,80.0bp)  .. (213.0bp,80.0bp) .. controls (219.0bp,80.0bp) and (225.0bp,86.0bp)  .. (225.0bp,92.0bp) .. controls (225.0bp,92.0bp) and (225.0bp,192.0bp)  .. (225.0bp,192.0bp) .. controls (225.0bp,198.0bp) and (219.0bp,204.0bp)  .. (213.0bp,204.0bp) .. controls (213.0bp,204.0bp) and (20.0bp,204.0bp)  .. (20.0bp,204.0bp) .. controls (14.0bp,204.0bp) and (8.0bp,198.0bp)  .. (8.0bp,192.0bp) .. controls (8.0bp,192.0bp) and (8.0bp,92.0bp)  .. (8.0bp,92.0bp) .. controls (8.0bp,86.0bp) and (14.0bp,80.0bp)  .. (20.0bp,80.0bp);
\end{scope}
  \node (Core) at (336.0bp,106.0bp) [draw,draw=none] {$\KnCore$};
  \node (Bool) at (335.0bp,250.0bp) [draw,draw=none] {$\KnBool$};
  \node (CoreBox) at (285.0bp,34.0bp) [draw,draw=none] {$\KnCoreBox$};
  \node (Krom) at (177.0bp,178.0bp) [draw,draw=none] {$\KnKrom$};
  \node (HornBox) at (500.0bp,106.0bp) [draw,draw=none] {$\KnHornBox$};
  \node (Horn) at (494.0bp,178.0bp) [draw,draw=none] {$\KnHorn$};
  \node (KromDia) at (172.0bp,106.0bp) [draw,draw=none] {$\KnKromDia$};
  \node (HornDia) at (606.0bp,106.0bp) [draw,draw=none] {$\KnHornDia$};
  \node (KromBox) at (63.0bp,106.0bp) [draw,draw=none] {$\KnKromBox$};
  \node (CoreDia) at (389.0bp,34.0bp) [draw,draw=none] {$\KnCoreDia$};
  \draw [->,dashed] (Krom) ..controls (175.21bp,151.98bp) and (174.55bp,142.71bp)  .. (KromDia);
  \draw [red,->,solid] (Horn) ..controls (496.14bp,151.98bp) and (496.94bp,142.71bp)  .. (HornBox);
  \draw [red,->,solid] (Horn) ..controls (535.91bp,150.81bp) and (553.92bp,139.55bp)  .. (HornDia);
  \draw [red,->,solid] (KromDia) ..controls (215.2bp,85.116bp) and (222.26bp,82.327bp)  .. (229.0bp,80.0bp) .. controls (273.42bp,64.667bp) and (288.24bp,65.91bp)  .. (CoreDia);
  \draw [->,dashed] (HornDia) ..controls (566.65bp,85.093bp) and (560.19bp,82.312bp)  .. (554.0bp,80.0bp) .. controls (517.07bp,66.216bp) and (474.06bp,54.626bp)  .. (CoreDia);
  \draw [->,dashed] (Bool) ..controls (279.09bp,224.23bp) and (246.9bp,209.97bp)  .. (Krom);
  \draw [red,->,solid] (KromBox) ..controls (104.67bp,85.136bp) and (111.49bp,82.341bp)  .. (118.0bp,80.0bp) .. controls (155.13bp,66.659bp) and (198.15bp,55.157bp)  .. (CoreBox);
  \draw [->,dashed] (Krom) ..controls (234.17bp,151.83bp) and (263.35bp,138.99bp)  .. (Core);
  \draw [->,dashed] (Horn) ..controls (438.57bp,152.44bp) and (408.8bp,139.25bp)  .. (Core);
  \draw [red,->,solid] (Core) ..controls (317.4bp,79.474bp) and (310.12bp,69.483bp)  .. (CoreBox);
  \draw [->,dashed] (Bool) ..controls (391.5bp,224.12bp) and (425.12bp,209.32bp)  .. (Horn);
  \draw [->,dashed] (Krom) ..controls (134.21bp,150.72bp) and (115.67bp,139.34bp)  .. (KromBox);
  \draw [red,->,solid] (Core) ..controls (355.39bp,79.389bp) and (363.05bp,69.277bp)  .. (CoreDia);
  \draw [->,dashed] (HornBox) ..controls (456.57bp,85.077bp) and (449.63bp,82.322bp)  .. (443.0bp,80.0bp) .. controls (400.18bp,64.998bp) and (386.05bp,65.476bp)  .. (CoreBox);
%
%
	\node (Equiv) at (33.0bp,185.0bp) [red] {$\equiv$};
	\draw [->,dashed] (10.0bp,250.0bp) -- (60.0bp,250.0bp) node[right] {``is weakly more expressive"};
	\draw [red,->,solid] (10.0bp,230.0bp) -- (60.0bp,230.0bp) node[right, black] {``is more expressive"};
\draw  plot [smooth] coordinates {(280bp,260bp) (400bp,160bp) (650bp,160bp) };
\node[rotate=2.5] at (600bp,170bp) {$\PSpace$-hard~\cite{ChenLin94}};
\end{tikzpicture}
	\caption{An account of the results of this paper.}
   \label{fig:exprpower}
\end{figure} 

\bibliographystyle{eptcs}
\bibliography{biblio}

\begin{thebibliography}{10}
\providecommand{\bibitemdeclare}[2]{}
\providecommand{\surnamestart}{}
\providecommand{\surnameend}{}
\providecommand{\urlprefix}{Available at }
\providecommand{\url}[1]{\texttt{#1}}
\providecommand{\href}[2]{\texttt{#2}}
\providecommand{\urlalt}[2]{\href{#1}{#2}}
\providecommand{\doi}[1]{doi:\urlalt{http://dx.doi.org/#1}{#1}}
\providecommand{\bibinfo}[2]{#2}

\bibitemdeclare{book}{spatial_logic_handbook}
\bibitem{spatial_logic_handbook}
\bibinfo{editor}{M.~\surnamestart Aiello\surnameend},
  \bibinfo{editor}{I.~\surnamestart Pratt-Hartmann\surnameend} \&
  \bibinfo{editor}{J.~\surnamestart {van Benthem}\surnameend}, editors
  (\bibinfo{year}{2007}): \emph{\bibinfo{title}{Handbook of Spatial Logics}}.
\newblock \bibinfo{publisher}{Springer}, \doi{10.1007/978-1-4020-5587-4}.

\bibitemdeclare{inproceedings}{DBLP:conf/lpar/ArtaleKRZ13}
\bibitem{DBLP:conf/lpar/ArtaleKRZ13}
\bibinfo{author}{A.~\surnamestart Artale\surnameend},
  \bibinfo{author}{R.~\surnamestart Kontchakov\surnameend},
  \bibinfo{author}{V.~\surnamestart Ryzhikov\surnameend} \&
  \bibinfo{author}{M.~\surnamestart Zakharyaschev\surnameend}
  (\bibinfo{year}{2013}): \emph{\bibinfo{title}{The Complexity of Clausal
  Fragments of {LTL}}}.
\newblock In: {\sl \bibinfo{booktitle}{Proc. of LPAR 2013}}, {\sl
  \bibinfo{series}{LNCS}} \bibinfo{volume}{8312},
  \bibinfo{publisher}{Springer}, pp. \bibinfo{pages}{35--52},
  \doi{10.1007/978-3-642-45221-5\_3}.

\bibitemdeclare{article}{DBLP:journals/tocl/ArtaleKRZ14}
\bibitem{DBLP:journals/tocl/ArtaleKRZ14}
\bibinfo{author}{A.~\surnamestart Artale\surnameend},
  \bibinfo{author}{R.~\surnamestart Kontchakov\surnameend},
  \bibinfo{author}{V.~\surnamestart Ryzhikov\surnameend} \&
  \bibinfo{author}{M.~\surnamestart Zakharyaschev\surnameend}
  (\bibinfo{year}{2014}): \emph{\bibinfo{title}{A Cookbook for Temporal
  Conceptual Data Modelling with Description Logics}}.
\newblock {\sl \bibinfo{journal}{{ACM} Transactions on Computational Logic}}
  \bibinfo{volume}{15}(\bibinfo{number}{3}), pp. \bibinfo{pages}{1--50},
  \doi{10.1145/2629565}.

\bibitemdeclare{inproceedings}{artale2015}
\bibitem{artale2015}
\bibinfo{author}{A.~\surnamestart Artale\surnameend},
  \bibinfo{author}{R.~\surnamestart Kontchakov\surnameend},
  \bibinfo{author}{V.~\surnamestart Ryzhikov\surnameend} \&
  \bibinfo{author}{M.~\surnamestart Zakharyaschev\surnameend}
  (\bibinfo{year}{2015}): \emph{\bibinfo{title}{Tractable Interval Temporal
  Propositional and Description Logics}}.
\newblock In: {\sl \bibinfo{booktitle}{Proc. of AAAI 2015}},
  \bibinfo{publisher}{{AAAI} Press}, pp. \bibinfo{pages}{1417--1423}.

\bibitemdeclare{article}{asp79}
\bibitem{asp79}
\bibinfo{author}{B.~\surnamestart Aspvall\surnameend}, \bibinfo{author}{M.~F.
  \surnamestart Plass\surnameend} \& \bibinfo{author}{R.~E. \surnamestart
  Tarjan\surnameend} (\bibinfo{year}{1979}): \emph{\bibinfo{title}{A Linear
  Time Algorithm for Testing the Truth of Certain Quantified Boolean
  Formulas}}.
\newblock {\sl \bibinfo{journal}{Information Processing Letters}}
  \bibinfo{volume}{8}(\bibinfo{number}{3}), pp. \bibinfo{pages}{121--123},
  \doi{10.1016/0020-0190(79)90002-4}.

\bibitemdeclare{book}{Baader:2003:DLH:885746}
\bibitem{Baader:2003:DLH:885746}
\bibinfo{editor}{F.~\surnamestart Baader\surnameend},
  \bibinfo{editor}{D.~\surnamestart Calvanese\surnameend},
  \bibinfo{editor}{D.L. \surnamestart McGuinness\surnameend},
  \bibinfo{editor}{D.~\surnamestart Nardi\surnameend} \& \bibinfo{editor}{P.F.
  \surnamestart Patel-Schneider\surnameend}, editors (\bibinfo{year}{2003}):
  \emph{\bibinfo{title}{The Description Logic Handbook: Theory, Implementation,
  and Applications}}.
\newblock \bibinfo{publisher}{Cambridge University Press},
  \doi{10.2277/0521781760}.

\bibitemdeclare{book}{modal-logic}
\bibitem{modal-logic}
\bibinfo{author}{P.~\surnamestart Blackburn\surnameend},
  \bibinfo{author}{M.~\surnamestart de~Rijke\surnameend} \&
  \bibinfo{author}{Y.~\surnamestart Venema\surnameend} (\bibinfo{year}{2002}):
  \emph{\bibinfo{title}{Modal Logic}}.
\newblock \bibinfo{publisher}{Cambridge University Press}.

\bibitemdeclare{article}{DBLP:journals/jcss/BloemJPPS12}
\bibitem{DBLP:journals/jcss/BloemJPPS12}
\bibinfo{author}{R.~\surnamestart Bloem\surnameend},
  \bibinfo{author}{B.~\surnamestart Jobstmann\surnameend},
  \bibinfo{author}{N.~\surnamestart Piterman\surnameend},
  \bibinfo{author}{A.~\surnamestart Pnueli\surnameend} \&
  \bibinfo{author}{Y.~\surnamestart Sa'ar\surnameend} (\bibinfo{year}{2012}):
  \emph{\bibinfo{title}{Synthesis of Reactive(1) Designs}}.
\newblock {\sl \bibinfo{journal}{Journal of Computer and System Sciences}}
  \bibinfo{volume}{78}(\bibinfo{number}{3}), pp. \bibinfo{pages}{911--938},
  \doi{10.1016/j.jcss.2011.08.007}.

\bibitemdeclare{unpublished}{bresolin:2016}
\bibitem{bresolin:2016}
\bibinfo{author}{D.~\surnamestart Bresolin\surnameend},
  \bibinfo{author}{A.~\surnamestart Kurucz\surnameend},
  \bibinfo{author}{E.~\surnamestart Mu{\~n}oz-Velasco\surnameend},
  \bibinfo{author}{V.~\surnamestart Ryzhikov\surnameend},
  \bibinfo{author}{G.~\surnamestart Sciavicco\surnameend} \&
  \bibinfo{author}{M.~\surnamestart Zakharyaschev\surnameend}
  (\bibinfo{year}{2016}): \emph{\bibinfo{title}{Horn Fragments of the
  Halpern-Shoham Interval Temporal Logic (Technical Report)}}.
\newblock \urlprefix\url{http://arxiv.org/abs/1604.03515v1}.
\newblock \bibinfo{note}{Preliminary version}.

\bibitemdeclare{inproceedings}{DBLP:conf/jelia/BresolinMS14}
\bibitem{DBLP:conf/jelia/BresolinMS14}
\bibinfo{author}{D.~\surnamestart Bresolin\surnameend},
  \bibinfo{author}{E.~\surnamestart Mu{\~{n}}oz{-}Velasco\surnameend} \&
  \bibinfo{author}{G.~\surnamestart Sciavicco\surnameend}
  (\bibinfo{year}{2014}): \emph{\bibinfo{title}{Sub-propositional Fragments of
  the Interval Temporal Logic of {A}llen's Relations}}.
\newblock In: {\sl \bibinfo{booktitle}{Proc. of JELIA 2014}}, {\sl
  \bibinfo{series}{LNCS}} \bibinfo{volume}{8761},
  \bibinfo{publisher}{Springer}, pp. \bibinfo{pages}{122--136},
  \doi{10.1007/978-3-319-11558-0\_9}.

\bibitemdeclare{article}{ChenLin93}
\bibitem{ChenLin93}
\bibinfo{author}{C.C. \surnamestart Chen\surnameend} \& \bibinfo{author}{I.P.
  \surnamestart Lin\surnameend} (\bibinfo{year}{1993}):
  \emph{\bibinfo{title}{The computational complexity of satisfiability of
  temporal {Horn} formulas in propositional linear-time temporal logic}}.
\newblock {\sl \bibinfo{journal}{Information Processing Letters}}
  \bibinfo{volume}{45}(\bibinfo{number}{3}), pp. \bibinfo{pages}{131--136},
  \doi{10.1016/0020-0190(93)90014-Z}.

\bibitemdeclare{article}{ChenLin94}
\bibitem{ChenLin94}
\bibinfo{author}{C.C. \surnamestart Chen\surnameend} \& \bibinfo{author}{I.P.
  \surnamestart Lin\surnameend} (\bibinfo{year}{1994}):
  \emph{\bibinfo{title}{The Computational Complexity of the Satisfiability of
  Modal Horn Clauses for Modal Propositional Logics}}.
\newblock {\sl \bibinfo{journal}{Theoretical Computer Science}}
  \bibinfo{volume}{129}(\bibinfo{number}{1}), pp. \bibinfo{pages}{95--121},
  \doi{10.1016/0304-3975(94)90082-5}.

\bibitemdeclare{book}{cookhorn}
\bibitem{cookhorn}
\bibinfo{author}{S.~\surnamestart Cook\surnameend} \&
  \bibinfo{author}{P.~\surnamestart Nguyen\surnameend} (\bibinfo{year}{2010}):
  \emph{\bibinfo{title}{Logical foundations of proof complexity}}.
\newblock \bibinfo{publisher}{Cambridge University Press},
  \doi{10.1017/CBO9780511676277}.

\bibitemdeclare{article}{del1987note}
\bibitem{del1987note}
\bibinfo{author}{L.~\surnamestart {Fari{\~n}as Del Cerro}\surnameend} \&
  \bibinfo{author}{M.~\surnamestart Penttonen\surnameend}
  (\bibinfo{year}{1987}): \emph{\bibinfo{title}{A note on the complexity of the
  satisfiability of modal {H}orn clauses}}.
\newblock {\sl \bibinfo{journal}{Journal of Logic Programming}}
  \bibinfo{volume}{4}(\bibinfo{number}{1}), pp. \bibinfo{pages}{1--10},
  \doi{10.1016/0743-1066(87)90018-5}.

\bibitemdeclare{book}{temporal_logic_foundations}
\bibitem{temporal_logic_foundations}
\bibinfo{author}{D.~\surnamestart Gabbay\surnameend},
  \bibinfo{author}{I.~\surnamestart Hodkinson\surnameend} \&
  \bibinfo{author}{M.~\surnamestart Reynolds\surnameend}
  (\bibinfo{year}{1994}): \emph{\bibinfo{title}{Temporal Logic: mathematical
  foundations and computational aspects, Volume 1}}.
\newblock \bibinfo{series}{Oxford Logic Guides}, \bibinfo{publisher}{Oxford
  University Press}.

\bibitemdeclare{article}{interval_modal_logic}
\bibitem{interval_modal_logic}
\bibinfo{author}{J.Y. \surnamestart Halpern\surnameend} \&
  \bibinfo{author}{Y.~\surnamestart Shoham\surnameend} (\bibinfo{year}{1991}):
  \emph{\bibinfo{title}{A Propositional Modal Logic of Time Intervals}}.
\newblock {\sl \bibinfo{journal}{Journal of the ACM}} \bibinfo{volume}{38}, pp.
  \bibinfo{pages}{279--292}, \doi{10.1145/115234.115351}.

\bibitemdeclare{article}{horn}
\bibitem{horn}
\bibinfo{author}{A.~\surnamestart Horn\surnameend} (\bibinfo{year}{1951}):
  \emph{\bibinfo{title}{On Sentences Which Are True of Direct Unions of
  Algebras}}.
\newblock {\sl \bibinfo{journal}{Journal of Symbolic Logic}}
  \bibinfo{volume}{16}(\bibinfo{number}{1}), pp. \bibinfo{pages}{14--21},
  \doi{10.2307/2268661}.

\bibitemdeclare{inproceedings}{DBLP:conf/dlog/KonevLWZ15}
\bibitem{DBLP:conf/dlog/KonevLWZ15}
\bibinfo{author}{B.~\surnamestart Konev\surnameend},
  \bibinfo{author}{C.~\surnamestart Lutz\surnameend},
  \bibinfo{author}{F.~\surnamestart Wolter\surnameend} \&
  \bibinfo{author}{M.~\surnamestart Zakharyaschev\surnameend}
  (\bibinfo{year}{2015}): \emph{\bibinfo{title}{Conservative Rewritability of
  Description Logic TBoxes: First Results}}.
\newblock In: {\sl \bibinfo{booktitle}{Proc. of the 28th DL Workshop}}, {\sl
  \bibinfo{series}{CEUR-WS.org}} \bibinfo{volume}{1350}, pp.
  \bibinfo{pages}{196 -- 207}.

\bibitemdeclare{article}{krom}
\bibitem{krom}
\bibinfo{author}{M.R. \surnamestart Krom\surnameend} (\bibinfo{year}{1970}):
  \emph{\bibinfo{title}{The Decision Problem for Formulas in Prenex Conjunctive
  Normal Form with Binary Disjunction}}.
\newblock {\sl \bibinfo{journal}{Journal of Symbolic Logic}}
  \bibinfo{volume}{35}(\bibinfo{number}{2}), pp. \bibinfo{pages}{14--21},
  \doi{10.2307/2270511}.

\bibitemdeclare{article}{nguyen2004complexity}
\bibitem{nguyen2004complexity}
\bibinfo{author}{L.A. \surnamestart Nguyen\surnameend} (\bibinfo{year}{2004}):
  \emph{\bibinfo{title}{On the complexity of fragments of modal logics}}.
\newblock {\sl \bibinfo{journal}{Advances in Modal Logic}} \bibinfo{volume}{5},
  pp. \bibinfo{pages}{318--330}.

\bibitemdeclare{book}{papa}
\bibitem{papa}
\bibinfo{author}{C.H. \surnamestart Papadimitriou\surnameend}
  (\bibinfo{year}{1994}): \emph{\bibinfo{title}{Computational Complexity}}.
\newblock \bibinfo{publisher}{Addison Wesley}.

\end{thebibliography}

\end{document}